\documentclass[a4paper,11pt]{article}

\usepackage[utf8x]{inputenc}
\usepackage{amsmath}
\usepackage{amsfonts}
\usepackage{amssymb}
\usepackage{graphicx}

\newtheorem{prop}{Proposition}
\newtheorem{theo}{Theorem}
\newtheorem{conj}{Conjecture}
\newtheorem{claim}{Claim}

\newtheorem{prob}{Problem}

\newenvironment{proof}{\noindent \textbf{Proof}}{\rule{2mm}{2mm}}

\title{On packing time-respecting arborescences}
\author{Romain Chapoulli\'e and Zolt\'an Szigeti, \ \ \ \ \ \ \\
Ensimag and GSCOP, \\Grenoble}

\date{}

\begin{document}

\maketitle

\begin{abstract} 
We present a slight generalization of the  result of Kamiyama and Kawase \cite{kamkaw} on packing time-respecting arborescences in acyclic pre-flow temporal networks. 
Our main contribution is to provide the first results on packing time-respecting arborescences in non-acyclic temporal networks.
As negative results, we prove the NP-completeness of the decision problem of the existence of 2 arc-disjoint spanning time-respecting arborescences and of a related problem  proposed in this paper.
\end{abstract}

\section{Introduction}

Temporal networks were introduced to model the exchange of information in a network or the spread of a disease in a population. We are given a directed graph $D$ and a time label function $\tau$ on the arcs of $D,$ the pair $(D,\tau)$ is called a temporal network. Intuitively, for an arc $a$ of $D$, $\tau(a)$ is the time when the end-vertices of $a$ communicate, that is when the tail of $a$ can transmit a piece of information to the head of $a.$ Then the information can propagate through a path $P$ if it is time-respecting, meaning that the time labels of the arcs of $P$ in the order they are passed are non-decreasing. For a nice introduction to temporal networks, see \cite{keklku}. 
\medskip

Problems about packing arborescences in temporal networks were investigated in \cite{kamkaw}. An arborescence is called time-respecting if all the directed paths it contains are time-respecting. The main result of \cite{kamkaw} provides a packing of time-respecting arborescences, each vertex belonging to many of them, if the network is pre-flow and acyclic. Here pre-flow means intuitively that each vertex different from the root has at least as many arcs entering as leaving, while acyclic means that no directed cycle exists. Kamiyama and Kawase  \cite{kamkaw} presented examples to show that these conditions can not be dropped. 
\medskip

Two  questions naturally arise from these results: 
Must all kinds of directed cycles be forbidden? 
Does high time-respecting root-connectivity imply the existence  of 2 arc-disjoint spanning time-respecting  arborescences in a non-pre-flow temporal network?
\medskip

Let us now present our contributions that give an answer to those questions. 
\medskip

We first propose a generalized version of the result of \cite{kamkaw}  with a simplified proof in Theorem \ref{theotauarb}.  
\medskip

Our main result, Theorem \ref{cyclic}, is about packing time-respecting arborescences in pre-flow temporal networks that may contain directed cycles. The condition in Theorem \ref{cyclic} is that the arcs in the same strongly connected component must  have the same $\tau$-value. 
If this condition holds then our intuition would be to use regular arborescences in the strongly
connected components and then to try to extend them to obtain a packing of time-respecting
arborescences in the temporal network. This idea is a step in the right direction, however the
exact process used in the proof is a bit more complex, see Section \ref{nonac}.
\medskip

By the famous result of Edmonds  \cite{ed}, we know that $k$-root-connectivity  implies the existence of a packing of $k$ spanning $s$-arborescences. The authors of \cite{keklku} show that for any positive integer $k$, time-respecting $k$-root-connectivity does not imply the existence of 2 arc-disjoint spanning time-respecting arborescences in a temporal network. To explain this  construction (or more precisely, a slightly modified version of it), we point out and recall in Section \ref{adstra} the close relation between packings of spanning time-respecting arborescences, packings of Steiner arborescences and proper 2-colorings of hypergraphs. We remark in Theorem \ref{SPNP} that the decision problem, whether there exist  2 arc-disjoint spanning time-respecting arborescences, is NP-complete.
\medskip

We show in Theorem \ref{newresult} that time-respecting $(n-1)$-root-connectivity implies the existence of a packing of $2$ spanning time-respecting $s$-arborescences in an arbitrary temporal network on $n$ vertices. This result becomes more interesting if we note that the examples of Figure \ref{circuitexists1} show that time-respecting $(n-3)$-root-connectivity is not enough.
\medskip

Finally, in Theorem \ref{NPC}, we  show that in an acyclic temporal network $(D,\tau)$, it is NP-complete to decide whether there exists a spanning arborescence whose   directed paths consist of arcs of the same $\tau$-value.

\section{Definitions}

Let {\boldmath $D$} $=(V\cup s,A)$ be a directed graph with a special vertex {\boldmath$s$}, called {\it root}, such that no arc enters $s$.
The set of arcs entering, leaving a vertex set $X$ of $D$ is denoted by {\boldmath $\rho_D(X)$}, {\boldmath $\delta_D(X)$}, respectively. Sometimes we use $\rho_A(X)$ for $\rho_D(X)$ and similarly $\delta_A(X)$ for $\delta_D(X)$. We denote $|\rho_D(X)|$ and $|\delta_D(X)|$ by {\boldmath$d^-_D(X)$} and {\boldmath$d^+_D(X)$}, respectively.
We call the directed graph $D$  {\it acyclic} if $D$ contains no directed cycle. 
If $d^-_D(v)=d^+_D(v)$ for all $v\in V$, then $D$ is called {\it Eulerian}.
We say that  $D$ is {\it pre-flow}  if 	$d^-_D(v)\ge d^+_D(v)$ for all $v\in V$. 
A subgraph $F=(V'\cup s,A')$ of $D$ is called an {\it $s$-arborescence} if $F$ is acyclic and $d^-_F(v)=1$ for all $v\in V'.$
We say that $F$ is {\it spanning} if $V'=V.$ For $U\subseteq V$,  $F$  is called a {\it Steiner $s$-arborescence} or an {\it $(s,U)$-arborescence} if $F$ is an $s$-arborescence and it contains all the vertices in $U.$ 
A {\it packing} of arborescences means a set of arc-disjoint arborescences. 
For $v\in V,$  a path from $s$ to $v$ is called an {\it $(s,v)$-path} and {\boldmath $\lambda_D(s,v)$} denotes the maximum number of arc-disjoint $(s,v)$-paths in $D.$ For some $k\in\mathbb{N}$, we say that $D$ is {\it $k$-root-connected} if $\lambda_D(s,v)\ge k$ for all $v\in V.$ For some $U\subseteq V$ and $k\in\mathbb{N}$, we say that $D$ is {\it Steiner $k$-root-connected} if $\lambda_D(s,v)\ge k$ for all $v\in U.$ 
We call a directed graph $D'=(V\cup \{s,t\},A')$ {\it almost Eulerian} if $d^-_{D'}(v)=d^+_{D'}(v)$ for all $v\in V$ and  $d^-_{D'}(s)=0=d^+_{D'}(t).$
\medskip

For a function {\boldmath $\tau$} $: A\rightarrow \mathbb{N}$,  {\boldmath $N$} $=(D,\tau)$ is called a {\it temporal network}.
For $i\in \mathbb{N}$, let 
{\boldmath $\rho_N^i(v)$}$:=\{a\in \rho_D(v):\tau(a)\le i\}$ and 
{\boldmath $\delta_N^i(v)$}$:=\{a\in \delta_D(v):\tau(a)\le i\}$. 
We call the temporal network $N$  {\it acyclic} if $D$ is acyclic. We say that $N$ is {\it pre-flow}  if $|\rho_N^i(v)|\ge |\delta_N^i(v)|$ for all $i\in \mathbb{N}$ and for all  $v\in V.$ 
Note that if a temporal network $(D,\tau)$ is pre-flow, then the directed graph $D$ is pre-flow. 
We say that $(D,\tau)$ is {\it consistent} if arcs of different $\tau$-values cannot belong to the same  strongly connected component of $D$. In this case in each strongly connected component $Q$ of $D$ that contains at least one arc, each arc has the same $\tau$-value, that we denote by {\boldmath $\tau(Q)$}.
A directed path $P$ of $D$, consisting of the arcs $a_1,\dots, a_\ell$ in this order, is called {\it time-respecting} or {\it $\tau$-respecting} if $\tau (a_i)\le \tau(a_{i+1})$ for $1\le i\le \ell-1.$ An $s$-arborescence $F$ of $D$ is called {\it time-respecting} or {\it $\tau$-respecting} if for every vertex $v$ of $F$, the unique $(s,v)$-path in $F$ is $\tau$-respecting. For $v\in V,$ {\boldmath $\lambda_N(s,v)$} denotes the maximum number of arc-disjoint $\tau$-respecting $(s,v)$-paths in $D.$ We say that $N$ is {\it time-respecting $k$-root-connected} if $\lambda_N(s,v)\ge k$ for all $v\in V.$ If $N'=(D',\tau')$ is a temporal network where $D'=(V\cup \{s,t\},A')$ is  almost Eulerian, then for a vertex $v\in V,$ we call a bijection $\mu'_v$ from $\delta_{D'}(v)$ to $\rho_{D'}(v)$  {\it $\tau'$-respecting} if   $\tau'(\mu'_v(f))\le \tau'(f)$ for all $f\in \delta_{D'}(v).$
\medskip

A hypergraph {\boldmath$\mathcal{H}$} $=(V,\mathcal{E})$ is defined by its vertex set $V$ and its hyperedge set $\mathcal{E}$ where a hyperedge is a subset of $V.$ For some $r\in \mathbb{N},$ the hypergraph $\mathcal{H}$ is called {\it $r$-uniform} if each hyperedge in $\mathcal{E}$ is of size $r$ and {\it $r$-regular} if each vertex in $V$ belongs to exactly $r$ hyperedges. 
A $2$-coloring of the vertex set $V$ is called {\it proper} if each hyperedge in $\mathcal{E}$ contains vertices of both colors, in other words no monochromatic hyperedge exists in $\mathcal{E}.$ We call $\mathcal{E}'\subseteq\mathcal{E}$ an {\it exact cover} of $\mathcal{H}$ if each vertex in $V$ belongs to exactly one hyperedge in $\mathcal{E}'.$

\section{Packing time-respecting arborescences in acyclic pre-flow temporal networks}

The aim of this section is to generalize the following result of Kamiyama and Kawase \cite{kamkaw} on packing time-respecting arborescences in acyclic pre-flow temporal networks.

\begin{theo}[\cite{kamkaw}]\label{kaka}
Let $N=((V\cup s, A),\tau)$ be an acyclic pre-flow temporal network and  $k\in \mathbb{N}.$
There exists a packing of $k$ $\tau$-respecting $s$-arborescences such that each vertex $v$ in $V$ belongs to  $\min\{k,\lambda_N(s,v)\}$ of them.
\end{theo}

Note that Theorem \ref{kaka} implies that in a time-respecting $k$-root-connected acyclic pre-flow temporal network there exists a packing of $k$ spanning time-respecting $s$-arborescences.
\medskip

 We now present our first result, a slight extension of Theorem \ref{kaka}.

\begin{theo}\label{theotauarb}
Let $N=((V\cup s, A),\tau)$ be an acyclic temporal network and $k\in \mathbb{N}$ such that  
\begin{equation}\label{ourcond}
	\min\{k,|\rho_N^i(v)|\}\ge \min\{k,|\delta_N^i(v)|\} \ \ \ \text{ for all } i\in \mathbb{N}, \text{ for all }  v\in V. 
\end{equation}
There exists a packing of $k$ $\tau$-respecting $s$-arborescences  such that each vertex $v$ in $V$ belongs to  $\min\{k,d^-_A(v)\}$ of them.
\end{theo}

We will partially follow the proof of \cite{kamkaw} but we will point out that Lemmas 3 and 4 in  \cite{kamkaw} are not needed to prove Theorem \ref{theotauarb}. Hence the proof of Theorem \ref{theotauarb} is simpler than that of Theorem \ref{kaka}. The following algorithm is a slightly modified version of the algorithm of Kamiyama and Kawase \cite{kamkaw}. 
Its input is an acyclic  temporal network $N=((V\cup s, A),\tau)$ and $k\in \mathbb{N}$ such that  \eqref{ourcond} is satisfied.
Its output is a packing of $\tau$-respecting $s$-arborescences $T_1,\dots,T_k$ such that each vertex $v$ in $V$ belongs to  $\min\{k,d^-_A(v)\}$ of them.
For every $v\in V,$ let  {\boldmath$I(v)$} be a set of arcs of smallest $\tau$-values entering $v$ of size $\min\{k,d^-_A(v)\}$.
The algorithm will use arcs only in $\bigcup _{v\in V}I(v).$
The algorithm heavily relies on the fact that the network is acyclic. It is well-known that a directed graph $D$ is acyclic if and only if a {\it topological ordering} $v_1,\dots,v_n$ of its vertex set exists, that is if $v_iv_j$ is an arc of $D$ then $i>j.$ 
Since no arc enters $s$, we may suppose that in a topological ordering $v_n=s.$
\medskip

\noindent {\bf Algorithm} {\sc "Packing Time-Respecting Arborescences"}
\smallskip

Let {\boldmath$v_n$} $=s,\dots,$ {\boldmath$v_1$} be a topological ordering of $V\cup s$.

Let {\boldmath $A_i$} $=\emptyset$ for all $1\le i\le k.$

For $j=1$ to $n-1,$ let

\hskip .7truecm  {\boldmath $I$} $=\{1\le i\le k: \delta _{A_i}(v_j)\neq\emptyset\},$

\hskip .7truecm   {\boldmath $a_i$} be an arc in $\delta _{A_i}(v_j)$ of minimum $\tau$-value for all $i\in I,$

\hskip .7truecm  $\{${\boldmath $\bar a_1$}$,\dots,${\boldmath$\bar a_{|I|}$}$\}$  be an ordering of $\{a_{i}:i\in I\}$ such that $\tau(\bar a_1)\le\dots\le\tau(\bar a_{|I|})$,

\hskip .7truecm   {\boldmath $\pi$} $:I \rightarrow\{1,\dots, |I|\}$ be the bijection such that $a_i=\bar a_{\pi(i)}$ for all $i\in I,$

\hskip .7truecm  {\boldmath $J$} be a subset of $\{1,\dots, k\}\setminus I$ of size $|I(v_j)|-|I|,$

\hskip .7truecm   {\boldmath $\sigma$} $:J\rightarrow\{1,\dots, |J|\}$ be a bijection,

\hskip .7truecm  $\{${\boldmath$e_1,\dots,e_{|I|},f_1,\dots,f_{|J|}$}$\}$ be an ordering of $I(v_j)$ such that 

\hskip 1.7truecm $\tau(e_1)\le\dots\le\tau(e_{|I|})\le\tau(f_1)\le\dots\le\tau(f_{|J|}),$

\hskip .7truecm   {\boldmath $A_i$} $=A_i\cup e_{\pi(i)}$ for all $i\in I,$ 

\hskip .7truecm   {\boldmath $A_i$} $=A_i\cup f_{\sigma(i)}$  for all $i\in J.$

Let  {\boldmath $T_i$} $=(V_i,A_i)$ where {\boldmath $V_i$} is the vertex set of the arc set $A_i$ for all $1\le i\le k.$

Stop.

\begin{theo}\label{justification}
Given an acyclic  temporal network $N=((V\cup s, A),\tau)$ and $k\in \mathbb{N}$ such that  \eqref{ourcond} is satisfied, 
Algorithm {\sc Packing Time-Respecting Arborescences} outputs a packing of $k$ $\tau$-respecting $s$-arborescences such that each vertex $v$ in $V$ belongs to  $\min\{k,d^-_A(v)\}$ of them.
\end{theo}

\begin{proof}
For all $1\le j\le n-1$, in the $j^{th}$ iteration of the algorithm, by the definition of $I$, \eqref{ourcond} and the definition of $I(v_j)$, we have $|I|\le\min\{k, d^+_A(v_j)\}\le\min\{k, d^-_A(v_j)\}=|I(v_j)|.$ This implies that  $J$ exists.
By construction, the digraphs $T_1,\dots, T_k$ are pairwise arc-disjoint and the in-degree of each vertex $v_j\in V_i-s$ is $1$ in $T_i.$ Then, since $N$ is acyclic, $T_i$ is an $s$-arborescence for all $1\le i\le k$. 
Moreover, $|\{1\le i\le k:v_j\in V_i\}|=|I|+|J|=|I(v_j)|=\min\{k, d^-_A(v_j)\}$ for all $1\le j\le n-1.$
To see that $T_i$ is time-respecting for all $1\le i\le k$, let $v_j$ be a vertex in $V_i-s$ and $a\in \delta_{A_i}(v_j)$. Then $e_{\pi(i)}\in \rho_{A_i}(v_j)$. Suppose on the contrary that $\tau(e_{\pi(i)})>\tau(a).$ 
Since $\tau(g)\ge \tau(e_{\pi(i)})>\tau(a)$ for all $g\in \rho_A(v_j)\setminus \{e_1,\dots,e_{{\pi(i)}-1}\}$, we have $|\rho_N^{\tau(a)}(v_j)|\le |\{e_1,\dots, e_{{\pi(i)}-1}\}|={\pi(i)}-1.$ Since $\tau(a)\ge \tau(a_i)=\tau(\bar a_{\pi(i)})\ge \tau(\bar a_\ell)$ for all $1\le\ell\le \pi(i)$ and $\pi(i)\le|I|\le k$, we have $\pi(i)=|\{\bar a_1,\dots, \bar a_{\pi(i)}\}|\le\min\{|\delta_N^{\tau(a)}(v_j)|,k\}$. Thus $|\rho_N^{\tau(a)}(v_j)|<\min\{|\delta_N^{\tau(a)}(v_j)|,k\}$ that contradicts \eqref{ourcond}.  This contradiction completes the proof.
\end{proof}
\medskip

Note that Theorem \ref{justification} implies Theorem \ref{theotauarb}. Note also that Theorem \ref{theotauarb} implies Theorem \ref{kaka}. Indeed, if $N$ is pre-flow, then \eqref{ourcond} is satisfied, so, by Theorem \ref{theotauarb}, there exists a packing of $k$ $\tau$-respecting $s$-arborescences  such that each vertex $v$ in $V$ belongs to exactly $\min\{k,d^-_A(v)\}$ of them. This implies that $\min\{k,\lambda_N(s,v)\}=\min\{k,d^-_A(v)\}$ and hence Theorem \ref{kaka} follows.

\section{Packing time-respecting arborescences in non-acyclic pre-flow temporal networks}\label{nonac}

In \cite{kamkaw}, Kamiyama and Kawase  provide an example of 7 vertices and 12 arcs that shows that in Theorem \ref{kaka} one can not delete the condition that $D$ is acyclic. Here we provide a smaller example with 5 vertices and 7 arcs, see the first temporal network in Figure \ref{circuitexists1}. Note that this temporal network contains a directed cycle whose arcs  are not of the same $\tau$-values and hence the temporal network is not consistent.
\medskip

\begin{figure}[h]
	\centerline{\includegraphics[scale=.5]{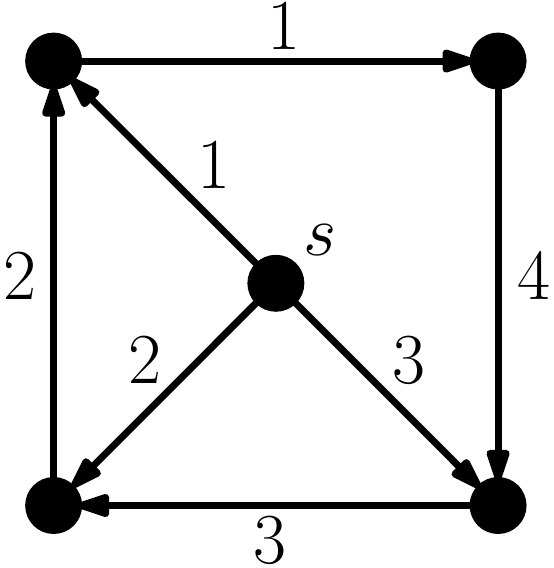}
	\hskip 3truecm
	\includegraphics[scale=.5]{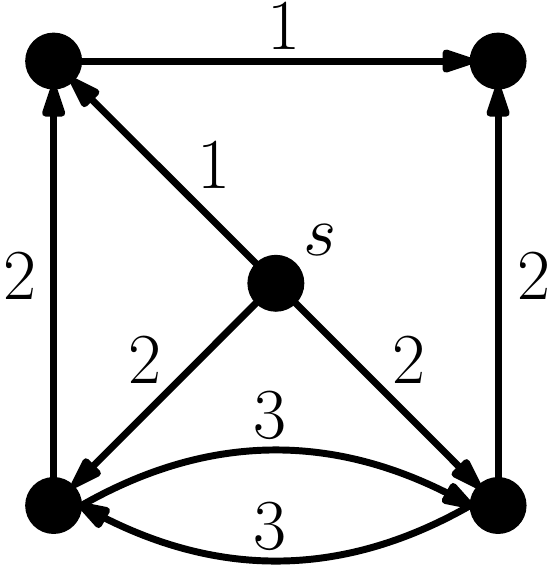}
	\hskip 3truecm
	\includegraphics[scale=.5]{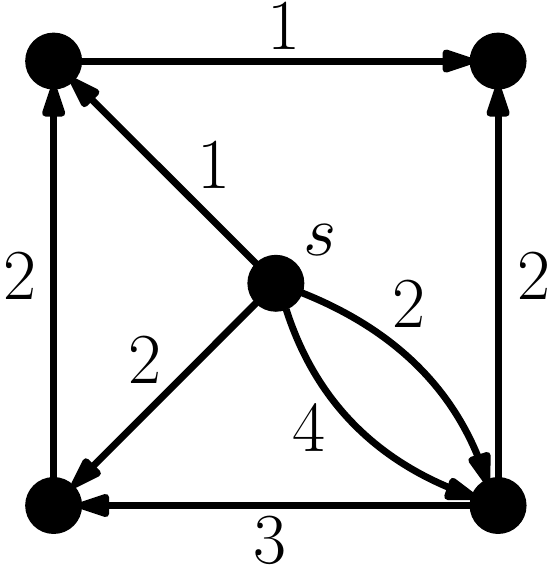}}
	\caption{Three temporal networks $N$ where the $\tau$-value of an arc is presented on the arc. The first two are non-acyclic pre-flow, the second one is consistent. The third one is acyclic but not pre-flow. They contain no 2 arc-disjoint $\tau$-respecting $s$-arborescences  such that each vertex $v$  belongs to  $\min\{2,\lambda_N(s,v)\}$ of them.}\label{circuitexists1}
\end{figure}

The second temporal network in Figure \ref{circuitexists1} is another example that shows that in Theorem \ref{kaka} one can not delete the condition that $D$ is acyclic. Here the temporal network contains one directed cycle $C$ and all the arcs of $C$ are of the same $\tau$-values and hence the temporal network is  consistent.
Note that in this  example  there exists a packing of three  $\tau$-respecting $s$-arborescences  such that each vertex $v$  belongs to exactly $\lambda_N(s,v)$ of them. 
\medskip

Kamiyama and Kawase \cite{kamkaw} also provide an example of 7 vertices and 12 arcs that shows that in Theorem \ref{kaka} one can not delete the condition that $D$ is pre-flow. Here we provide a smaller example with 5 vertices and 8 arcs, see the third temporal network in Figure \ref{circuitexists1}. 
\medskip

We now present the main result of this paper on packing of time-respecting arborescences in consistent pre-flow temporal networks where only the natural upper bound is given on the number of arborescences.  

\begin{theo}\label{cyclic}
Let $N=(D=(V\cup s,A),\tau)$ be a consistent pre-flow temporal network. 
There exists a packing of $d^+_D(s)$ $\tau$-respecting $s$-arborescences, each vertex $v$ in $V$ belonging to  $\lambda_N(s,v)$ of them.
\end{theo}

To prove Theorem \ref{cyclic}, we need an easy observation on almost Eulerian acyclic pre-flow temporal networks.
A similar result has already been presented in \cite{kamkaw}. 

\begin{prop}\label{almeul}
If $N=(D=(V\cup\{s,t\},A),\tau)$ is an almost Eulerian acyclic  temporal network and $\mu_v$ is a $\tau$-respecting bijection from $\delta_D(v)$ to $\rho_D(v)$ for all $v\in V$, then $D$ decomposes into $d^+_D(s)$ $\tau$-respecting $(s,t)$-paths such that each vertex $v\in V$ belongs to $d^-_D(v)$ of them.
\end{prop}

\begin{proof}
We prove the claim by induction on $d^+_D(s)$. If $d^+_D(s)=0$, then, since $D$ is almost Eulerian and acyclic, we have $d^-_D(v)=0$ for all $v\in V$ and we are done. Otherwise, there exists an arc leaving $s.$ Then, using the bijections $\mu^{-1}_v$ and the facts  that $D$ is acyclic and $\mu_v$ is a $\tau$-respecting, we find a $\tau$-respecting directed $(s,t)$-path $P.$ By deleting the arcs of $P$ and applying the induction, the claim follows.
\end{proof}
\medskip

We also need the following result of Bang-Jensen, Frank, Jackson \cite{bfj}.  

\begin{theo}[\cite{bfj}]\label{bjfj}
Let $D=(V\cup s,A)$ be a pre-flow directed graph.   
There exists a packing of  $s$-arborescences, each vertex $v\in V$ belonging to  $\lambda_D(s,v)$ of them.
\end{theo}

We are  ready to prove our main result.
\medskip

\begin{proof} {\bf (of Theorem \ref{cyclic})}
First we transform the instance into another one {\boldmath $N'$} $=(D',\tau')$ as follows. The directed graph {\boldmath $D'$} $=(V\cup \{s,t\}, A\cup A')$ is obtained from $D$ by adding a new vertex {\boldmath $t$} and $d^-_D(v)-d^+_D(v)$ parallel arcs from $v$ to $t$ for all $v\in V$ and we define {\boldmath $\tau'$}$(a)$ to be equal to $\tau(a)$ if $a\in A$ and to $M$ if $a\in A'$, where {\boldmath $M$} $=\max\{\tau(a): a\in A\}.$ Since $N$ is pre-flow, so is $D$, that is $d^-_D(v)-d^+_D(v)\ge 0$ for all $v\in V$ and hence the construction is correct.
This way we get an instance which remains consistent ($\{t\}$ is a new strongly connected component) and pre-flow (by the definition of $M$) and $D'$ is  almost Eulerian.
\medskip

For each vertex $v\in V$, let us fix  orderings of $\rho_{D'}(v)$ and  $\delta_{D'}(v)$ such that $\tau'(e_1)\le\dots\le\tau'(e_{d^-_{D'}(v)})$ and $\tau'(f_1)\le\dots\le\tau'(f_{d^+_{D'}(v)})$, respectively. Then $\mu'_v(f_j)=e_j$ for all $1\le j\le d^+_{D'}(v)$ is a $\tau'$-respecting bijection for all $v\in V.$
Indeed, if there exists $j$ such that $\tau'(e_j)=\tau'(\mu'_v(f_j))>\tau'(f_j)=:i,$ then $|\rho_{N'}^i(v)|\le j-1<j\le|\delta_{N'}^i(v)|$ that contradicts the fact that $N'$ is pre-flow.
\medskip

To reduce the problem to an easy acyclic problem that can be treated by Proposition \ref{almeul} and some problems that can be treated by Theorem \ref{bjfj}, let us denote the strongly connected components of $D'$  by {\boldmath $Q'_1$}$,\dots,$ {\boldmath $Q'_{\ell}$}. 
Let {\boldmath $U_j$} denote the vertex set of $Q_j'$ for all $1\le j\le \ell.$
Then the directed graph {\boldmath $D''$} obtained from $D'$ by contracting each $Q'_j$ into a vertex {\boldmath $q''_j$}   is acyclic. By changing the indices if it is necessary, we may suppose that $q''_{\ell}=s, \dots, q''_1=t$ is a topological ordering of the vertices of $D''.$ 
Let   {\boldmath $N''$} $=(D'',\tau'')$ be the temporal network where {\boldmath $\tau''$}$(a)=\tau'(a)$ for all $a\in A(D'').$   Note  that since $D'$ is almost Eulerian, so is $D''.$ Indeed, we have $d^-_{D''}(q''_j)-d^+_{D''}(q''_j)=d^-_{D'}(U_j)-d^+_{D'}(U_j)=\sum_{v\in U_j}(d^-_{D'}(v)-d^+_{D'}(v))=0$ for all $2\le j\le \ell-1.$ Note also that $d^+_D(s)=d^+_{D'}(s)=d^+_{D''}(s).$
\medskip

To define  a convenient $\tau''$-respecting bijection $\mu_j''$ from $\delta_{D''}(q''_j)=\delta_{D'}(U_j)$ to $\rho_{D''}(q''_j)=\rho_{D'}(U_j)$ for all $2\le j\le \ell-1$,  let us fix such a {\boldmath $j$} and let us define the following sets:
\medskip

 {\boldmath $R _j^1$} $=\{vw\in \delta_{D'}(U_j): \tau'(\mu'_v(vw))>\tau'(Q'_j)\}$, 
 
 {\boldmath $R _j^2$} $=\{vw\in \delta_{D'}(U_j): \tau'(vw)<\tau'(Q'_j)\}$, 
 
 {\boldmath $R _j^3$} $=\delta_{D'}(U_j)\setminus(R _j^1 \cup R _j^2)$, 
 
 {\boldmath $S _j^1 $} $=\{\mu'_v(vw): vw\in R _j^1 \}$, 
 
 {\boldmath $S _j^2$} $=\{\mu'_v(vw): vw\in R _j^2\}$ and 
 
 {\boldmath $S _j^3$} $=\rho_{D'}(U_j)\setminus(S _j^1 \cup S _j^2)$. 

\begin{claim}\label{partitions}
$\{R _j^1 ,R _j^2,R _j^3\}$ is a partition of $\delta_{D'}(U_j)$ and $\{S _j^1 ,S _j^2,S _j^3\}$ is a partition of $\rho_{D'}(U_j)$.
\end{claim}

\begin{proof}
If $vw\in R _j^1 ,$ $v'w'\in R _j^2,$ $uv=\mu'_v(vw)\in S _j^1 $ and $u'v'=\mu'_{v'}(v'w')\in S _j^2$, then, since  $\mu'_v$ and $\mu'_{v'}$ are  $\tau'$-respecting bijections, we have $\tau'(vw)\ge\tau'(\mu'_v(vw))=\tau'(uv)>\tau'(Q'_j)>\tau'(v'w')$ $\ge \tau'(\mu'_v(v'w'))=\tau'(u'v')$. Thus $vw\ne v'w'$ and $uv\ne u'v'$, so $R _j^1 \cap R _j^2=\emptyset$ and $S _j^1 \cap S _j^2=\emptyset.$ By the definition of $R _j^1$ and $R _j^2$, we have $R _j^1 \cup R _j^2\subseteq \delta_{D'}(U_j)$. If $vw\in R _j^1$, then $\tau'(\mu'_v(vw))>\tau'(Q'_j).$ If $vw\in R _j^2$, then, since $\mu'_v$ is a $\tau'$-respecting bijection, we get $\tau'(\mu'_v(vw))\le \tau'(vw)<\tau'(Q'_j).$ Then, using that each arc in $Q'_j$ has $\tau'$-value  $\tau'(Q'_j),$ we have $S _j^1 \cup S _j^2\subseteq \rho_{D'}(U_j)$. By the definition of $R _j^3$ and $S _j^3$, Claim \ref{partitions} follows.
\end{proof}
\medskip

We now start to define $\mu''_j$. For $vw\in R _j^1 \cup R _j^2$, let {\boldmath $\mu''_j$}$(vw)=\mu'_v(vw)$. 
Since each $\mu'_v$ is $\tau'$-respecting, we have $\tau''(vw)=\tau'(vw)\ge\tau'(\mu'_v(vw))=\tau''(\mu''_v(vw)).$
Note that for all $xy\in R _j^3$ and  for all $uv\in S _j^3$, $\tau'(xy)\ge \tau'(Q'_j)\ge\tau'(uv).$
However, we cannot take an arbitrary bijection from $R _j^3$ to $S _j^3$ because  we  have to guarantee that the vertices in $Q'_j$ also belong to the required number of arborescences. In order to do this, let us define the temporal network {\boldmath $N'_j$} $=(D'_j,\tau'_j)$ where    the directed graph {\boldmath $D'_j$} is obtained from $D'$ by contracting $\bigcup_{i>j} U_i$ into a vertex {\boldmath $s_j$}, contracting $\bigcup_{i<j} U_i$ into a vertex {\boldmath $t_j$} and deleting the arcs from $s_j$ to $t_j$ and {\boldmath $\tau'_j$}$(a)=\tau'(a)$ for all $a\in A(D'_j).$ 

\begin{claim}\label{lg} $N'_j$ satisfies the following.
\begin{itemize}
	\item [(a)] $D'_j$ is  almost Eulerian,
	\item [(b)] $\lambda_{D'_j}(s_j,t_j)=d^-_{D'_j}(t_j)$, 
	\item [(c)] $\lambda_{N'_j}(s_j,v)\ge \lambda_{N'}(s,v)$ for all $v\in U_j.$
\end{itemize}
\end{claim}

\begin{proof}
(a) Since $D'$ is  almost Eulerian, so is ${D'_j}$. Indeed, we have $d^-_{D'_j}(v)=d^-_{D'}(v)=d^+_{D'}(v)=d^+_{D'_j}(v)$ for all $v\in U_j.$ 

(b) By (a)  and $d^-_{D'_j}(s_j)=0=d^+_{D'_j}(t_j)$, (b) easily follows. Indeed, let {\boldmath$r_j$} $=d^-_{D'_j}(t_j)$ and let us define {\boldmath$D^*_j$} by adding $r_j$  arcs $\{${\boldmath$h_1$}$,\dots, ${\boldmath$h_{r_j}$}$\}$ from $t_j$ to $s_j$ in $D'_j$. Then, by (a), $D^*_j$  is Eulerian. Thus  it decomposes into directed cycles. Let {\boldmath$C_1, \dots, C_{r_j}$} be the arc-disjoint directed cycles that contain the arcs $h_1,\dots, h_{r_j}$. Then {\boldmath$P_1$} $=C_1-h_1, \dots,$ {\boldmath$P_{r_j}$} $=C_{r_j}-h_{r_j}$ are arc-disjoint directed $(s_j,t_j)$-paths. Hence ${r_j}\le\lambda_{D'_j}(s_j,t_j)\le{r_j}$, and we have (b).

(c) For all $v\in U_j,$ any $\tau'$-respecting $(s,v)$-path in $N'$ provides a $\tau'_j$-respecting $(s_j,v)$-path in $N'_j,$ and (c) follows.
\end{proof}
\medskip

To be able to use normal arborescences (not time-respecting ones), we have to modify $D'_j.$ No $\tau$-respecting directed path in $D$ may contain an arc in $S _j^1 $ and an arc in $Q'_j$, hence the corresponding arcs  in $R_j^1$ and $S_j^1$ will be deleted from  $D'_j$. A $\tau$-respecting $s$-arborescence in $D$ may contain an arc $\mu'_v(vw)$ in $S _j^2$ (where $vw\in R _j^2$) and an arc in $Q'_j$, but this arborescence must contain $vw$. To guarantee this property we use a trick: we replace the corresponding two arcs in $R_j^2$ and $S_j^2$ in $D'_j$ by two convenient arcs. More precisely, let {\boldmath $H_j$} be obtained from $D'_j$ by deleting $s_jv$ and $vt_j$ that correspond to $\mu'_v(vw)$ and $vw$ for all  $vw\in R _j^1 $ and replacing $s_jv$ and $vt_j$ that correspond to $\mu'_v(vw)$ and $vw$ for all  $vw\in R _j^2$ by {\boldmath$e_{vw}$} $=s_jt_j$ and {\boldmath$f_{vw}$} $=t_jv.$ Let {\boldmath$E_j$} $=\{e_{vw}:vw\in R _j^2\}$ 
and {\boldmath$F_j$} $=\{f_{vw}:vw\in R _j^2\}$.

\begin{claim}\label{lambda} 
$H_j$ satisfies the following.
\begin{itemize}
	\item [(a)]  $H_j$ is pre-flow,
	\item [(b)] $\lambda_{H_j}(s_j,t_j)=d^-_{H_j}(t_j)$, 
	\item [(c)]  $\lambda_{H_j}(s_j,v)\ge\lambda_{N'_j}(s_j,v)-d^-_{S_j^1}(v)$ for all $v\in U_j$.
\end{itemize}
\end{claim}

\begin{proof}
(a) By Claim \ref{lg}(a), $D'_j$ is  almost Eulerian. Then, by $\delta_{D'_j}(t_j)=\emptyset,$ $D'_j$ is pre-flow. By deleting from $D'_j$ the arcs $s_jv$ and $vt_j$ that correspond to $\mu'_v(vw)$ and $vw$ for all  $vw\in R _j^1 $, we decreased the in-degree and the out-degree of each vertex by the same value so the directed graph obtained this way remained pre-flow. By replacing $s_jv$ and $vt_j$ that correspond to $\mu'_v(vw)$ and $vw$ for all  $vw\in R _j^2$ by $s_jt_j$ and $t_jv,$ we may decrease the out-degrees of the vertices in $Q'_j$ but the in-degrees remained unchanged. 
Further, $d^+_{H_j}(t_j)=d^+_{D'_j}(t_j)+|F_j|=|E_j|\le d^-_{H_j}(t_j).$
It follows that $H_j$ is pre-flow.

(b) Note that for all $t_j\in X\subseteq U_j\cup t_j,$ $d^-_{H_j}(X)=d^-_{D'_j}(X)-|R _j^1 |.$ Then, by Claim \ref{lg}(b), we have 
$d^-_{H_j}(t_j)\ge\lambda_{H_j}(s_j,t_j)\ge\lambda_{D'_j}(s_j,t_j)-|R _j^1 |=d^-_{D'_j}(t_j)-|R _j^1 |=d^-_{H_j}(t_j)$ and (b) follows. 

(c) On the one hand, by deleting the arcs corresponding to $\rho_{S_j^1}(v)$, we destroyed at most $d^-_{S_j^1}(v)$ $\tau'_j$-respecting $(s_j,v)$-paths in $N'_j$ and we did not destroy  a  $\tau'_j$-respecting $(s_j,u)$-path in $N'_j$ for $u\in U_j\setminus v$ because each arc in $Q'_j$ has $\tau'_j$-value  $\tau'_j(Q'_j)$ and each arc in $\rho_{S_j^1}(v)$ has $\tau'_j$-value strictly larger than $\tau'_j(Q'_j).$ On the other hand, if a $\tau'_j$-respecting $(s_j,u)$-path $P$ contains $s_jv$ (corresponding to $\mu'_v(vw)$ for some $vw\in R _j^2$) in $N'_j$ then $P-s_jv+e_{vw}+f_{vw}$ is a directed $(s_j,u)$-path in $H_j.$ These arguments imply (c).
\end{proof}
\medskip

By Claim \ref{lambda}(a) and Theorem \ref{bjfj},  there exists a packing ${\cal B}_j$ of $s_j$-arborescences {\boldmath$T_j^i$} in $H_j$, each vertex $v\in U_j\cup t_j$ belonging to  $\lambda_{H_j}(s_j,v)$ of them.  Let us choose such a packing {\boldmath${\cal B}_j$} that minimizes the size of the set {\boldmath$F_{{\cal B}_j}$} of the arcs $f_{vw}\in F_j$ such that an arborescence {\boldmath$T^{f_{vw}}_j$} in ${\cal B}_j$ contains $f_{vw}$ but not $e_{vw}$. 

\begin{claim}\label{sbt}
${\cal B}_j$ satisfies the following.
\begin{itemize}
	\item [(a)]  $d^+_{H_j}(s_j)=|{\cal B}_j|=d^-_{H_j}(t_j),$ 
	\item [(b)] $F_{{\cal B}_j}=\emptyset$,
	\item [(c)] $\{T^i_j-s_j-t_j:T^i_j\in {\cal B}_j\}$ is a packing of arborescences in $Q_j'$, each vertex $v\in U_j$ belonging to $\lambda_{H_j}(s_j,v)$ of them.
\end{itemize}
\end{claim}

\begin{proof}
(a) By Claim \ref{lambda}(b), $t_j$ belongs to $\lambda_{H_j}(s_j,t_j)=d^-_{H_j}(t_j)$ of the $s_j$-arborescences in ${\cal B}_j$. Thus  each arc entering $t_j$ belongs to some $s_j$-arborescence in ${\cal B}_j$ and $d^-_{H_j}(t_j)\le |{\cal B}_j|$. Moreover, by construction and since $D'_j$ is almost Eulerian, we have $d^-_{H_j}(t_j)=d^-_{D'_j}(t_j)-|R_j^1|=d^+_{D'_j}(s_j)-|S_j^1|=d^+_{H_j}(s_j)\ge |{\cal B}_j|,$ and (a) follows.

(b) Suppose that $F_{{\cal B}_j}\neq\emptyset$. Let {\boldmath$E_{{\cal B}_j}$} $=\{e_{vw}:f_{vw}\in F_{{\cal B}_j}\}.$ By (a), every $e_{vw}\in E_{{\cal B}_j}$ is contained in an $s_j$-arborescence {\boldmath$T^{e_{vw}}_j$} in ${\cal B}_j.$

First suppose that for some {\boldmath$e_{vw}$} $\in E_{{\cal B}_j},$ $T^{e_{vw}}_j$ contains only  the arc $e_{vw}$. 
Note that $T^{f_{vw}}_j-f_{vw}$ consists of an $s_j$-arborescence $T'_j$ and a $v$-arborescence $T''_j$. Let {\boldmath${{\cal B}'_j}$} be obtained from ${{\cal B}_j}$ by replacing $T^{f_{vw}}_j$ by $T'_j$ and $T^{e_{vw}}_j$ by $e_{vw}+f_{vw}+T''_j.$ Then ${{\cal B}'_j}$ is a packing of $s_j$-arborescences in $H_j$ such that  each vertex $v\in U_j\cup  t_j$ belongs to  $\lambda_{H_j}(s_j,v)$ of them. Moreover, $f_{vw}$ and $e_{vw}$ belong to the same $s_j$-arborescence in ${{\cal B}'_j}$, that is $|F_{{\cal B}'_j}|<|F_{{\cal B}_j}|$ and we have a contradiction.

We may hence suppose that for every $e_{vw}\in E_{{\cal B}_j},$ $T^{e_{vw}}_j$ contains another arc, so by (a), contains an arc in $F_{{\cal B}_j}.$ Let ${{\cal B}'_j}$ be the set of those $s_j$-arborescences in ${{\cal B}_j}$ that contain an arc of $F_{{\cal B}_j}.$ Then $|F_{{\cal B}_j}|=|E_{{\cal B}_j}|\le |{{\cal B}'_j}|\le |F_{{\cal B}_j}|.$ Hence we have equality everywhere. It follows that every $s_j$-arborescences in ${{\cal B}'_j}$ contains exactly one arc from both $F_{{\cal B}_j}$ and $E_{{\cal B}_j}.$ Then for every $f_{vw}\in F_{{\cal B}_j}$, $T^{f_{vw}}_j$ contains an arc $e_{v'w'}\in E_{{\cal B}_j}$. Let {\boldmath${{\cal B}''_j}$} be obtained from ${{\cal B}_j}$ by replacing $e_{v'w'}$  by $e_{vw}\in E_{{\cal B}_j}$ in $T^{f_{vw}}_j$ for every $f_{vw}\in F_{{\cal B}_j}$. Then ${{\cal B}''_j}$ is a packing of $s_j$-arborescences in $H_j$ such that  each vertex $v\in U_j\cup t_j$ belongs to  $\lambda_{H_j}(s_j,v)$ of them. Moreover, $F_{{\cal B}''_j}=\emptyset$ and we have a contradiction.

(c) follows from the definition of ${\cal B}_j$, (a) and (b).
\end{proof}
\medskip

We now finish the definition of $\mu''_j$. Let $vw\in R _j^3$. Then $vw$ corresponds in $H_j$ to an arc {\boldmath$g_{vw}$} $=vt_j$ entering $t_j$. By Claim \ref{sbt}(a), $g_{vw}$ belongs to an $s_j$-arborescence {\boldmath$T^{g_{vw}}_j$} in ${{\cal B}_j}$.
Let us define {\boldmath$\mu''_j$}$(vw)\in S _j^3$ to be the arc $xq''_j$ of $D''$ that corresponds to the arc $s_ju$ in $H_j$ of the unique $(s_j,t_j)$-path of $T^{g_{vw}}_j$. Then  $\tau''_j(vw)=\tau'_j(vw)\ge \tau'_j(Q'_j)\ge\tau'_j(xq''_j)=\tau''_j(\mu''_j(vw))$ for all $vw\in R _j^3$.
\medskip

By the definition of $\mu''_j$ and Claim \ref{partitions}, we have a $\tau''$-respecting bijection $\mu''_j$ from $\delta_{D''}(q''_j)$ to $\rho_{D''}(q''_j)$ for all $2\le j\le \ell-1.$ Recall that $D''$ is acyclic and almost Eulerian.
Then, by Proposition \ref{almeul} and $d^+_D(s)=d^+_{D''}(s)$, $D''$ decomposes into  $\tau''$-respecting $(s,t)$-paths {\boldmath $P_1,\dots,P_{d^+_D(s)}$} such that each vertex $q''_j\ne s$ belongs to $d^-_{D''}(q''_j)$ 
of them. These paths can be extended, using from Claim \ref{sbt}(c) the arborescences $T_j^i-s_j-t_j$ in $Q'_j$ for $1\le i\le d^+_{H_j}(s_j)$ and $2\le j\le \ell-1$, to get $s$-arborescences in $D'$ such that each vertex $v\in V$ belongs to $\lambda_{H_j}(s_j,v)+d^-_{S_j^1}(v)\ge\lambda_{N'_j}(s_j,v)\ge\lambda_{N'}(s,v)$ of them, by Claims \ref{lambda}(b) and \ref{lg}(c). Since the directed paths $P_1,\dots,P_{d^+_D(s)}$ are $\tau''$-respecting, that is $\tau'$-respecting and $D'$ is consistent, the arborescences constructed are $\tau'$-respecting. Hence $N'$ has a packing of $\tau'$-respecting $s$-arborescences {\boldmath $T'_1,\dots,T'_{d^+_D(s)}$} such that each vertex $v$ of $D'$ distinct from $s$ and $t$ belongs to  $\lambda_{N'}(s,v)=\lambda_{N}(s,v)$ of them, and hence 
$\{${\boldmath $T_1$} $=T'_1-t,\dots,$ {\boldmath $T_{d^+_D(s)}$} $=T'_{d^+_D(s)}-t\}$
is a packing of $\tau$-respecting $s$-arborescences such that each vertex $v$ of $D$ distinct from $s$ belongs to  $\lambda_{N}(s,v)$ of them.
\end{proof}

\section{Arc-disjoint spanning time-respecting arborescences}\label{adstra}

Edmonds' arborescence packing theorem  \cite{ed} states that $k$-root-connectivity from $s$ implies the existence of a packing of $k$ spanning $s$-arborescences. The following observation of  \cite{keklku} shows that the natural extension of Edmonds theorem for $k=1$ is true for temporal networks.

\begin{theo}[\cite{keklku}]\label{1arb}
Any $\tau$-respecting root-connected temporal network $N=((V\cup s,A),\tau)$ contains a spanning $\tau$-respecting $s$-arborescence. 
\end{theo}

The authors of \cite{keklku} show that high time-respecting root-connectivity of a temporal network does not imply the existence of 2 arc-disjoint spanning time-respecting arborescences. 

\begin{theo}[\cite{keklku}]\label{kextemp}
For all $k\in \mathbb{N}^+$, there exist  temporal networks $N=((V\cup s,A),\tau)$ such that $\lambda_{N}(s,v)\ge k$ for all $v\in V$ and no packing of $2$ spanning $\tau$-respecting $s$-arborescences exists in $N.$
\end{theo}

Their construction contains directed cycles but it can be easily modified to get an acyclic example. This acyclic example for $k=2$ is presented in Figure 2 in \cite{kamkaw}. 
\medskip

We now relate the spanning time-respecting arborescence packing problem to known problems, namely the Steiner arborescence packing problem and the hypergraph proper 2-coloring problem.
To do that we explain how the above mentioned modified construction can be obtained in 3 steps. 
First, take a $k$-uniform hypergraph without proper $2$-coloring. Then construct a directed graph that is Steiner $k$-root-connected without 2 arc-disjoint Steiner arborescences. Finally, construct an acyclic temporal network that is time-respecting $k$-root-connected without $2$ arc-disjoint spanning time-respecting arborescences. 
\medskip



There exist many constructions for $k$-uniform hypergraphs without proper $2$-coloring, see \cite{am}, \cite{keklku} and Exercice 13.45(b) of \cite{lovasz}. We mention that, by a result of Erd\H os \cite{erdos}, all examples contain exponentially many hyperedges in $k.$

\begin{theo}\cite{erdos}\label{exp}
Any $k$-uniform hypergraph without a proper $2$-coloring contains at least $2^{k-1}$ hyperedges.
\end{theo}



We now show that starting from an arbitrary $k$-uniform hypergraph $\mathcal{H}_k=(V_k, \mathcal{E}_k)$ without proper $2$-coloring how to construct an acyclic directed graph $D_k$ and a vertex set $U_k$ such that $\lambda_{D_k}(s,u)=k$ for all $u\in U_k$ and there exists no packing of two $(s,U_k)$-arborescences in $D_k$.
Let {\boldmath $G_k$} $:=(V_k, U_k; E_k)$ be the bipartite incidence graph of the hypergraph $\mathcal{H}_k$, where the elements of {\boldmath $U_k$} correspond to the hyperedges in $\mathcal{E}_k$. Let {\boldmath $D_k$} $=(V_k\cup U_k\cup s,A_k)$ be obtained from $G_k$ by adding a vertex $s$ and an arc $sv$ for all $v\in V_k$ and directing each edge of $E_k$ from $V_k$ to $U_k.$ By construction $D_k$ is acyclic. Since $\mathcal{H}_k$ is $k$-uniform, we have  $\lambda_{D_k}(s,u)= k$ for all $u\in U_k$.

\begin{theo}\label{2kex}
	$D_k$ has no packing of two $(s,U_k)$-arborescences. 
\end{theo}

\begin{proof}
Suppose that  there exists a packing of 2 $(s,U_k)$-arborescences $F_1$ and $F_2$ in $D_k.$ Using this packing, we can define a 2-coloring of $V_k$: let $v\in V_k$ be colored by $1$ if $sv\in A(F_1)$ and by $2$ otherwise. Since each vertex in $U_k$ belongs to both $F_1$ and $F_2$, no hyperedge of $\mathcal{E}_k$ is monochromatic, that is the  above defined $2$-coloring of $\mathcal{H}_k$ is proper. This contradicts the fact that $\mathcal{H}_k$ has no proper $2$-coloring.
\end{proof}
\medskip


As a next step, we show that starting from the acyclic directed graph $D_k$ and the vertex set $U_k$, how to construct a temporal network $N_k$ such that $\lambda_{N_k}(s,v)=k$ for all vertices $v$ and no packing of $2$ spanning time-respecting $s$-arborescences exists in $N.$
%
Let us define {\boldmath $N_k$} $:=(D_k^*,\tau_k^*)$ as follows: {\boldmath $D_k^*$} is obtained from $D_k$ by adding the set of arcs {\boldmath $A_k^*$} consisting of $k-1$ parallel arcs from $s$ to all $v\in V_k$ and we define {\boldmath $\tau_k^*$}$(a)=1$ if $a\in A_k$ and $2$ if $a\in A_k^*.$ Note that since $D_k$ is acyclic, so is $D_k^*$. Then a spanning $s$-arborescence $F^*$ of $D_k^*$ is $\tau_k^*$-respecting if and only if $F^*-A_k^*$ is an  $(s,U_k)$-arborescence in $D_k$. Thus a packing of $2$ spanning $\tau_k^*$-respecting $s$-arborescences in $D_k^*$ would provide a packing of $2$  $(s,U_k)$-arborescences in $D_k$. 
Hence, the following result is an immediate consequence of Theorem \ref{2kex}.

 \begin{theo}
For all $k\in \mathbb{N}^+$, there exist acyclic temporal networks $N=((V\cup s,A),\tau)$ such that $\lambda_{N}(s,v)\ge k$ for all $v\in V$ and no packing of $2$ spanning $\tau$-respecting $s$-arborescences exists in $N.$
\end{theo}


These examples of acyclic temporal networks that are time-respecting $k$-root-connected without $2$ arc-disjoint spanning time-respecting arborescences contain, by Theorem \ref{exp}, exponentially many vertices in $k.$ 
In other words, $k\le log(n)$ where $n$ is the number of vertices.
In the light of this fact, it is natural to ask whether there exist $2$ arc-disjoint spanning time-respecting arborescences in a temporal network if $k$ is linear in $n.$ The examples of Figure \ref{circuitexists1} show that time-respecting $(n-3)$-root-connectivity does not imply the existence of $2$ arc-disjoint spanning time-respecting arborescences.  We propose the first  steps in this direction. 
We  first remark that $n$-root-connectivity is enough.

\begin{claim}
Let $N=((V\cup s,A),\tau)$ be a temporal network on $n\ge 1$ vertices such that $\lambda_{N}(s,v)\ge n$ for all $v\in V$. Then there exists a packing of $2$ spanning $\tau$-respecting $s$-arborescences  in $N.$
\end{claim}

\begin{proof}
Since $\lambda_{N}(s,v)\ge n\ge 1$ for all $v\in V$, there exists, by Theorem \ref{1arb}, a spanning $\tau$-respecting $s$-arborescence {\boldmath$F$} in $N.$ Further, there exist $n$ arc-disjoint $\tau$-respecting $(s,v)$-paths {\boldmath$P_1^v,\dots, P_{n}^v$}  for all $v\in V$. By deleting the arcs of $F$, we can destroy at most $|A(F)|$ of the $(s,v)$-paths $P_1^v,\dots, P_{n}^v$ for all $v\in V.$ Since $|A(F)|=n-1,$ this implies that $\lambda_{N-A(F)}(s,v)\ge n-(n-1)=1$ for all $v\in V.$ Then, there exists, by Theorem \ref{1arb}, a spanning $\tau$-respecting $s$-arborescence {\boldmath$F'$} in $N-A(F),$ and we are done.
\end{proof}
\medskip

With some effort we can improve the previous result by 1.

\begin{theo}\label{newresult}
Let $N=((V\cup s,A),\tau)$ be a temporal network on $n\ge 2$ vertices such that $\lambda_{N}(s,v)\ge n-1$ for all $v\in V$. Then there exists a packing of $2$ spanning $\tau$-respecting $s$-arborescences  in $N.$
\end{theo}

\begin{proof}
Since $\lambda_{N}(s,v)\ge n-1\ge 1$ for all $v\in V$, there exists, by Theorem \ref{1arb}, a spanning $\tau$-respecting $s$-arborescence {\boldmath $F$} in $N.$ Let {\boldmath $F(v)$} be the unique arc of $F$ entering $v$ for all $v\in V.$ Note that $A(F)=\{F(v):v\in V\}.$
If $\lambda_{N-A(F)}(s,v)\ge 1$ for all $v\in V$ then there exists, by Theorem \ref{1arb}, a spanning $\tau$-respecting $s$-arborescence in $N-A(F),$ and we are done. 

Otherwise, $\lambda_{N-A(F)}(s,u)=0$ for some {\boldmath$u$} $\in V.$ 
By assumption, there exist $n-1$ arc-disjoint $\tau$-respecting $(s,u)$-paths {\boldmath$P_1,\dots, P_{n-1}$}.
Then, since $|V|=n-1,$ there exists a bijection {\boldmath$\pi$} from $V$ to $\{1,\dots,n-1\}$ such that $F(v)$ is contained in $P_{\pi(v)}$ for all $v\in V.$ It follows that no arc leaves $u$ in $F.$ Let {\boldmath$w$} $\in V-u$ be a vertex for which $\tau(F(w))$ is maximum. 
Let the last arc of $P_{\pi(w)}$ be denoted by  {\boldmath$xu$}.
 Then, since $F(u)$ is the last arc of the path $P_{\pi(u)}$ and the paths are arc-disjoint, $F(u)\neq xu.$ By the choice of $w$ and since $P_{\pi(w)}$ is $\tau$-respecting, we have $\tau(F(x))\le \tau(F(w))\le \tau(xu).$ We obtain that  {\boldmath$F'$} $:=F-F(u)+xu\neq F$ is also a spanning $\tau$-respecting $s$-arborescence  in $N.$ 
 
By assumption and  $|A(F)-F(u)|=n-2,$ we have $\lambda_{N-(A(F)-F(u))}(s,v)\ge (n-1)-(n-2)=1$ for all $v\in V.$ 
Then, by Theorem \ref{1arb}, there exists a spanning $\tau$-respecting $s$-arborescence {\boldmath$F''$} in $N-(A(F)-F(u)).$ Since $F''$ contains a unique arc entering $u$, it does not contain either $F(u)$ or $xu$. Thus, $F''$ is arc-disjoint from either $F$ or $F',$ and we are done.
\end{proof}
\medskip

We conjecture  that the following is true.

\begin{conj}
Let $N=((V\cup s,A),\tau)$ be an acyclic temporal network on $n\ge 4$ vertices such that $\lambda_{N}(s,v)\ge \frac n2$ for all $v\in V$. Then   a packing of $2$ spanning $\tau$-respecting $s$-arborescences exists in $N.$
\end{conj}

The third example  presented in Figure \ref{circuitexists1} is of $5$ vertices, acyclic, time-respecting $2$-root-connected and has no packing of $2$ spanning $\tau$-respecting $s$-arborescences. It follows that time-respecting $\frac {2n}{5}$-root-connectivity is not enough to have a packing of $2$ spanning time-respecting $s$-arborescences in acyclic temporal networks.

\section{Complexity results}

Lov\'asz \cite{lovasz2} proved that the problem of 2-colorings of $k$-uniform hypergraphs is NP-complete. This implies that the problem of packing 2 Steiner arborescences is also NP-complete. An easier way to see this is to use the NP-complete problem of two arc-disjoint directed paths in a directed graph $D$, one from $r$ to $t$ and the other from $t$ to $r.$ (See \cite {hota}.) Construct $D'$ from $D$ by adding a new vertex $s$ and the two arcs $sr$ and $st.$ Then $D$ has an $(r,t)$-path and a $(t,r)$-path that are arc-disjoint if and only if $D'$ has a packing of 2 $(s,\{r,t\})$-arborescences.
This with the construction presented in the previous section finally imply the following. 

\begin{theo}\label{SPNP}
The problem of packing $k$ spanning time-respecting arborescences is NP-complete even for $k=2.$
\end{theo}
\medskip

Let us check what happens if  we replace the inequality with equality in the definition of time-respecting directed paths and we consider the values of $\tau$ as colors. Then we get monochromatic directed paths. We may hence study the following problem {\sc MoChPaSpAr}: 

\begin{prob} 
Given a directed graph $D=(Z\cup s,A)$ and a coloring $c$ of the arcs, decide whether there exists a spanning $s$-arborescence containing only monochromatic directed paths.
\end{prob}

We show  that this decision problem is difficult. We will reduce the exact cover in 3-regular 3-uniform hypergraphs  problem ({\sc RXC3}) to our problem.  In {\sc RXC3}, we are given a 3-regular 3-uniform hypergraph $\mathcal{H}=(V,{\cal E})$, and the problem consists of determining whether there exists a subset ${\cal E}'$ of ${\cal E}$ such that each vertex in $V$ occurs in exactly one hyperedge in ${\cal E}'$. Gonzalez proved in \cite{gonz} that {\sc RXC3} is NP-complete.

\begin{theo}\label{NPC}
The problem {\sc MoChPaSpAr} is NP-complete even for acyclic directed graphs and for two colors.
\end{theo}

\begin{proof}
It is clear that {\sc MoChPaSpAr} is in NP. Let us take an instance of {\sc RXC3}, that is let $\mathcal{H}$ be a 3-regular 3-uniform hypergraph. We construct a polynomial size instance $(D,c)$ of {\sc MoChPaSpAr} such that $\mathcal{H}$ has an exact  cover if and only if $(D,c)$ has a spanning $s$-arborescence containing only monochromatic directed paths. 
Since $\mathcal{H}$ is a 3-regular 3-uniform hypergraph, the number of vertices of $\mathcal{H}$ and the number of hyperedges of $\mathcal{H}$ coincide.
Let us denote the vertices of $\mathcal{H}$ by $V=\{v_1,\dots, v_h\}$ and the hyperedges of $\mathcal{H}$ by  $\mathcal{E}=\{H_1,\dots, H_h\}$. 
\medskip

Let {\boldmath $D$} $=(Z\cup s, A)$ be the directed graph where {\boldmath $Z$} $=U\cup V\cup W$ and {\boldmath $A$} $=A_1\cup A_2\cup A_3\cup A_4$ with {\boldmath $U$} $=\{u_1,\dots, u_h\},$ {\boldmath $W$} $=\{w_{i,j}: H_i\cap H_j\neq\emptyset\}$,  {\boldmath $A_1$} $=\{e_i^1=su_i: 1\le i\le h\}$, {\boldmath $A_2$} $=\{e_i^2=su_i: 1\le i\le h\}$, {\boldmath $A_3$} $=\{u_iv_j: u_i\in U, v_j\in V, v_j\in H_i\}$ and {\boldmath $A_4$} $=\{u_iw_{i,j}, u_jw_{i,j}:u_i, u_j\in U, w_{i,j}\in W\}.$ Let {\boldmath $c$}$(a)$ be equal to black if $a\in A_1\cup A_3$  and grey if $a\in A_2\cup A_4$. Note that $D$ is acyclic and $c$ uses only two colors. For an example see Figure \ref{npred}. 

\begin{figure}[h]
\hskip 1truecm	{\includegraphics[scale=.5]{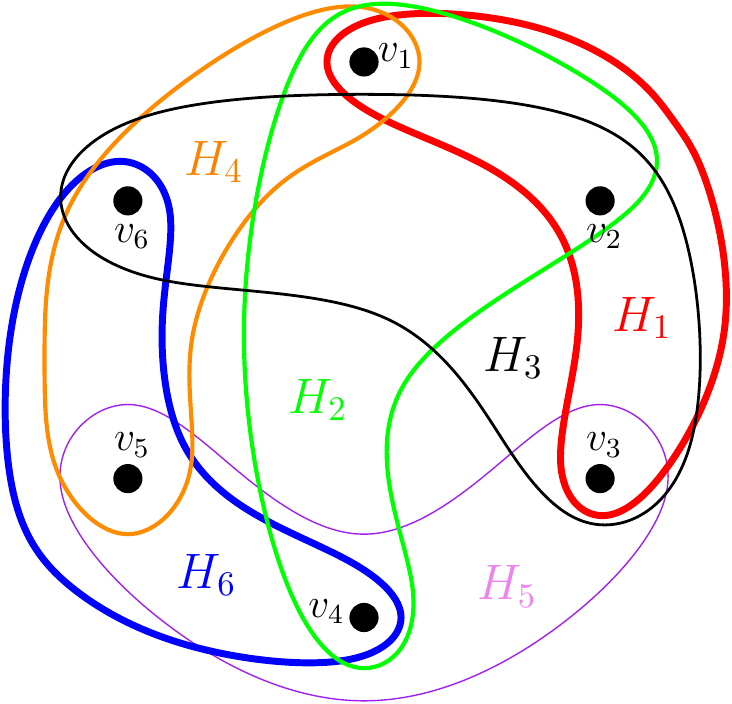}}
\hskip 2truecm	{\includegraphics[scale=.6]{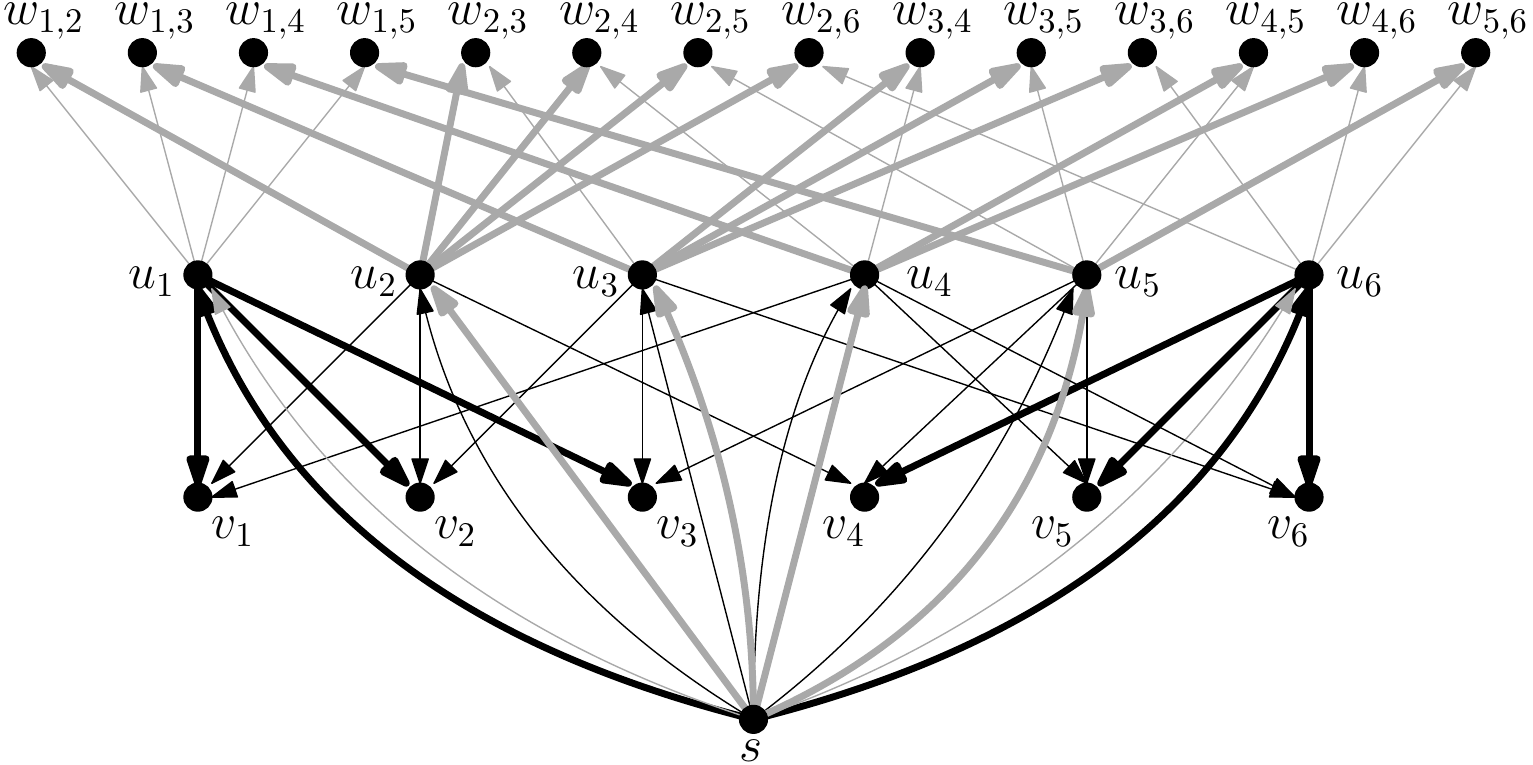}}
	\caption{A 3-regular 3-uniform hypergraph and the constructed  colored directed graph for it.}\label{npred}
\end{figure}
\medskip

The size of $D$ is polynomial in $h$. Indeed, since $\mathcal{H}$ is a 3-regular 3-uniform hypergraph, $|W|\le \frac12\cdot 3\cdot 2\cdot h$, so  $|Z\cup s|=|U|+|V|+|W|+1\le h+h+3h+1=5h+1$ and $|A|=|A_1|+|A_2|+|A_3|+|A_4|\le h+h+3h+2\cdot 3h=11h.$
\medskip

Suppose first that $\mathcal{H}$ has an exact  cover {\boldmath $\mathcal{H}'$}.
Let {\boldmath $Z'$} be the set of vertices of $D$ that can be reached from $s$ by a black directed path starting with an arc $su_i$ with $H_i\in \mathcal{H}'$ and {\boldmath $Z''$} by a grey directed path starting with an arc $su_i$ with $H_i\notin \mathcal{H}'$. Since $\mathcal{H}'$ is a cover, we have $Z'=V\cup \{u_i:H_i\in \mathcal{H}'\}.$ Since the hyperedges in $\mathcal{H}'$ are disjoint, we have $Z''=\{u_i :H_i\notin \mathcal{H}'\}\cup W.$ Since $Z'\cap Z''=s$, the desired spanning $s$-arborescence containing only monochromatic directed paths exists. In the example of Figure \ref{npred}, $\mathcal{H}'=\{H_1, H_6\},$ $Z'=V\cup \{u_1, u_6\}$, $Z''=\{u_2,u_3,u_4,u_5\}\cup W$ and the arborescence is represented by bold arcs.
\medskip

Now suppose that $(D,c)$ has a spanning $s$-arborescence {\boldmath $F$} containing only monochromatic directed paths. 
Let {\boldmath $\mathcal{H}'$} $=\{H_j:u_j\in U, v_i\in V, u_jv_i\in F\}$.
Since $F$ is a spanning $s$-arborescence, each vertex $v_i$ has exactly one black arc $u_jv_i$ in $F$ entering. This implies that $\mathcal{H}'$ covers $V.$ 
Let $H_j, H_k$ $(j<k)$ be hyperedges in $\mathcal{H}'.$ If $w_{j,k}\in W$, then, since the directed paths are monochromatic in $F$, $su_j$ and $su_k$ are black and hence  $u_jw_{j,k}$ and $u_kw_{j,k}$ are not contained in $F$ that contradicts the fact that  $F$ is a spanning $s$-arborescence. Thus $H_j$ and $H_k$ are disjoint. It follows that $\mathcal{H}'$ is an  exact  cover.
\end{proof}

\end{document}